\documentclass[12pt, a4paper]{article}

\usepackage{amsmath,amsfonts,amsthm,amssymb}
\usepackage[retainorgcmds]{IEEEtrantools}
\usepackage[ruled,vlined,linesnumbered]{algorithm2e}
\usepackage{tikz}
\usepackage{fourier}
\providecommand{\DontPrintSemicolon}{\dontprintsemicolon}

\theoremstyle{plain}
\newtheorem{theorem}{Theorem}[section]

\newtheorem{lemma}[theorem]{Lemma}
\newtheorem{proposition}[theorem]{Proposition}

\newtheorem{definition}[theorem]{Definition}
\newtheorem{corollary}[theorem]{Corollary}

\theoremstyle{remark}

\newtheorem{rem}[theorem]{Remark}

\newenvironment{sketchproof}
      {\medskip\noindent{\em Proof sketch.}\hspace{0.5ex}}
      {\hfill\qed\medskip}

\DeclareMathOperator{\opt}{OPT}

\newcommand{\NP}{{\mathrm{NP}}}

\DeclareMathOperator*{\argmax}{arg\,max}

\newcommand{\mysetminusD}{\hbox{\tikz{\draw[line width=0.6pt,line cap=round] (3pt,0) -- (0,6pt);}}}
\newcommand{\mysetminusT}{\mysetminusD}
\newcommand{\mysetminusS}{\hbox{\tikz{\draw[line width=0.45pt,line cap=round] (2pt,0) -- (0,4pt);}}}
\newcommand{\mysetminusSS}{\hbox{\tikz{\draw[line width=0.4pt,line cap=round] (1.5pt,0) -- (0,3pt);}}}
\newcommand{\mysetminus}{\mathbin{\mathchoice{\mysetminusD}{\mysetminusT}{\mysetminusS}{\mysetminusSS}}}
\allowdisplaybreaks

\title{Coverage, Matching, and Beyond:\\ New Results on Budgeted Mechanism Design\thanks{A conference version appears in WINE 2016. Research supported by an internal research funding program of the Athens University of Economics and Business.}}
\author{
Georgios Amanatidis 
\and
Georgios Birmpas 
\and
Evangelos Markakis \and 
\normalsize{Athens University of Economics and Business, Department of Informatics.} \\
\normalsize{\{gamana, gebirbas, markakis\}@aueb.gr}
}

\begin{document}
\maketitle

\begin{abstract}
We study a type of reverse (procurement) auction problems in the presence of budget constraints. The general algorithmic problem is to purchase a set of resources, which come at a cost, so as not to exceed a given budget and at the same time maximize a given valuation function. This framework captures the budgeted version of several well known optimization problems, and when the resources are owned by strategic agents the goal is to design truthful and budget feasible mechanisms, i.e.~elicit the true cost of the resources and ensure the payments of the mechanism do not exceed the budget. Budget feasibility introduces more challenges in mechanism design, and we study instantiations of this problem for certain classes of submodular and XOS valuation functions. We first obtain mechanisms with an improved approximation ratio for weighted {\em coverage valuations}, a special class of submodular functions that has already attracted attention in previous works. We then provide a general scheme for designing randomized and deterministic polynomial time mechanisms for a class of XOS problems. This class contains problems whose feasible set forms an {\em independence system} (a more general structure than matroids), and some representative problems include, among others, finding maximum weighted matchings, maximum weighted matroid members, and maximum weighted 3D-matchings. For most of these problems, only randomized mechanisms with very high approximation ratios were known prior to our results.
\end{abstract}

\newpage

\section{Introduction}
\label{sec:intro}

In this work, we study a class of mechanism design problems under a budget constraint. 
Consider a reverse auction setting, where a single buyer wants to select a subset, among a set $A$ of agents, for performing some tasks. Each agent $i$ comes at a cost $c_i$, in the case that he is chosen. The buyer has a budget $B$ and a valuation function $v(\cdot)$, so that $v(S)$ is the derived value if $S\subseteq A$ is the chosen set. The purely algorithmic version then asks to maximize the generated value subject to the constraint that the total cost of the selected agents should not exceed $B$ (often referred to  as a ``hard'' budget constraint). Some of these problems are motivated by crowdsourcing scenarios and related applications, where agents can be viewed as workers, e.g., \cite{AnariGN14}. Apart from that, they also form natural ``budgeted'' versions of well known optimization problems. 

In the mechanism design version that we focus on, the cost $c_i$ is private information for each agent $i$. Hence, we want to design mechanisms that are incentive compatible, individually rational, and budget feasible, i.e.~the sum of the payments to the agents does not exceed $B$. Note that the payments here can be higher than the actual costs in order to induce truthfulness. Budget feasibility is a tricky property that makes the problem more challenging, as it already rules out well known mechanisms such as VCG. Although the algorithmic versions of such problems often admit constant factor approximation algorithms, it is not clear how to appropriately convert them into truthful budget feasible mechanisms. Therefore, the question of interest is to find mechanisms that achieve the best possible approximation for the optimal value of $v(\cdot)$ under these constraints. We stress that the question is nontrivial even if we allow exponential time algorithms, since computational power does not necessarily make the problem easier (see also the discussion in \cite{DobzinskiPS11}).

Budgeted mechanism design was first studied by Singer \cite{Singer10} when $v(\cdot)$ is an additive or a nondecreasing submodular function. Later on, follow up works have also provided more results for XOS and subadditive functions (see the related work section). Although these results shed more light on our understanding of the problem, there are still several interesting issues that remain unresolved both for submodular and non-submodular cases. First, the current results on submodular valuations are not known to be tight. Further, and most importantly, when going beyond submodularity, to XOS functions, we are not even aware of general mechanisms with small approximation guarantees, let alone deterministic polynomial time mechanisms.  
\medskip

\noindent {\bf Contribution:} We first demonstrate (Section \ref{sec:submod}) how to obtain improved deterministic budget feasible mechanisms for weighted {\em coverage valuations}, a notable subclass of submodular functions. This class has already received attention in previous works \cite{Singer10,Singer12}, motivated by problems related to influence maximization in social networks. Our mechanism reduces roughly by half (from 31.03 to 15.45) the known approximation of \cite{Singer12} and also generalizes it to the weighted version of coverage functions. We then move to our main result (Section \ref{sec:nonsubmod}), which is a general scheme for obtaining randomized and deterministic polynomial time approximations for a subclass of XOS problems, that contains the budgeted versions of several well known optimization problems. We first illustrate our ideas in Section \ref{sub:matching}, on the budgeted matching problem, where $v(S)$ is defined as the maximum weight matching that can be derived from the edges of $S$. For this problem only a randomized 768-approximation was known~\cite{BeiCGL12}.Our approach yields a randomized $3$-approximation and a deterministic $4$-approximation. Then in Section \ref{sub:other}, we show how to generalize these results to problems with a similar combinatorial structure, where the set of feasible solutions forms an {\it independence system}. These structures are more general than matroids (they do not always satisfy the exchange property) and some representative problems that are captured include finding maximum weighted matroid members, maximum weighted $k$-D-matchings, and maximum weighted independent sets. For such problems we establish that a $\rho$-approximation to the algorithmic problem can be converted into a deterministic (resp. randomized), truthful, budget feasible mechanism with an approximation ratio of $2\rho +2$ (resp. $2\rho +1$). We conclude with some interesting open problems for future work.
\medskip

\noindent {\bf Related Work:} 
The study of budget feasible mechanisms, as considered here, was
initiated by Singer \cite{Singer10}, who gave a randomized constant factor approximation mechanism for nondecreasing submodular functions. Later, Chen et al.~\cite{ChenGL11} significantly improved these approximation ratios, obtaining a randomized, polynomial time mechanism achieving a $7.91$-approximation and a deterministic one with a $8.34$-approximation.
Their deterministic mechanism  does not run in polynomial time in general, but it can be modified to do so for special cases at the expense of its performance (see the beginning of Section \ref{sec:submod}). As an example, Singer \cite{Singer12} followed a similar approach to obtain a deterministic, polynomial time, $31.03$-approximation mechanism for the unweighted version of Budgeted Max Coverage, a class that we also consider in Section \ref{sec:submod}. Along these lines, Horel et al.~\cite{HorelIM14} consider another family of submodular functions and give a deterministic, polynomial time, constant approximation for the so-called Experimental Design Problem, under a mild relaxation on truthfulness.
For subadditive functions Dobzinski et al.~\cite{DobzinskiPS11} suggested a randomized $O(\log^2 n)$-approximation mechanism. This was later improved to $O\left(\log n / \log \log n\right)$ by Bei et al.~\cite{BeiCGL12}, who also gave a randomized $O(1)$-approximation mechanism for XOS functions, albeit in exponential time, and further initiated the Bayesian analysis in this setting. Recently, there is also a line of related work under the {\em large market} assumption (where no participant can affect significantly the market outcome). Under this assumption, Anari et al.~\cite{AnariGN14} resolved the additive case by giving a $\frac{e}{e-1}$-approximation mechanism and a matching lower bound. Further results for large markets were obtained by Goel et al.~\cite{GoelNS14} for a crowdsourcing problem with matching constraints (and hence a non submodular objective).

A related line of work involves the design of frugal mechanisms. These are mechanisms where one cares for minimizing the total amount of payments that are required by the mechanism, for finding a good solution. A series of results has been obtained over the years on designing frugal mechanisms, see e.g., \cite{ArcherT07,ChenEGP10,KarlinKT05,KempeSM10}. Frugal mechanism design is a complementary approach to budget feasibility, since here we have a hard budget constraint that should never be exceeded. Hence, results from this area do not generally transfer to our setting.

Finally, there is a plethora of works on auctions that take budgets into account, from the bidder's point of view, motivated mainly by sponsored search auctions, see among others, \cite{BorgsCIMS05,DobzinskiLN12,GoelML12} for some representative problems that have been tackled. Although these are fundamentally different problems than ours, they do highlight the difficulties that arise in the presence of budget constraints.

\section{Definitions and Notation}
\label{sec:defs}

We will use $A =[n]= \{1, 2, ..., n\}$ to denote a set of $n$ agents. Each agent $i$ is associated with a private cost $c_i$, denoting the cost for participating in the solution or for performing a certain task. We consider a procurement auction setting, where the auctioneer is equipped with a valuation function $v:2^{A}\to \mathbb{Q}^+$ and a positive budget $B$. Here $v(S)$, for $S\subseteq A$, is the value or happiness derived by the auctioneer if the set $S$ is selected. Therefore, the algorithmic goal in all the problems we study is to select a set $S$ that maximizes $v(S)$ subject to the constraint $\sum_{i\in S} c_i \leq B$.
 
We consider valuation functions that are non-decreasing, i.e.~$v(S)\le v(T)$ for any $S \subseteq  T \subseteq A$.
Throughout our work, we will focus on valuations that come from two natural classes of functions, namely submodular and XOS functions defined below. 

\begin{definition}
A valuation function, defined on $2^A$ for some set $A$, is
\begin{enumerate}
\item[\emph{(i)}] \emph{submodular}, if~~$v(S \cup \{i\}) - v(S) \geq v(T\cup \{i\}) - v(T)$ for any $S\subset T \subset A$, and $i\not\in T$.
\item[\emph{(ii)}] \emph{XOS} or \emph{fractionally subadditive}, if there exist additive functions $\alpha_1, ...,\alpha_r$, for some finite $r$, such that 
$v(S) = \max \{\alpha_1(S), \alpha_2(S), ...,\alpha_r(S) \}$.
\end{enumerate}
\end{definition}

\noindent We note that the class XOS is a strict superclass of submodular valuations. 

\medskip

\noindent{\bf Mechanism Design.} Each agent here only has his cost as private information, hence we are in the domain of single-parameter problems. A mechanism $\mathcal{M}=(f,p)$ in our context consists of an outcome rule $f$ and a payment rule $p$. Given a vector of cost declarations, 
$b = (b_i)_{i \in A}$, where $b_i$ denotes the cost reported by agent 
$i$, the mechanism selects the set $f(b)$. At the same time, it computes payments $p(b) = (p_i(b))_{i \in N}$ 
where $p_i(b)$ denotes the payment issued to agent $i$. 

The main properties we want to ensure for our mechanisms in this work are the following. 
\begin{definition}
A mechanism $\mathcal{M}=(f,p)$ is 
\begin{enumerate}
\item \emph{truthful}, if reporting $c_i$ is a dominant strategy for every agent $i$.
\item \emph{individually rational}, if~~$p_i(b)\geq 0$ for every $i\in A$, and $p_i(b) \geq c_i$, for every $i\in f(b)$.
\item  \emph{budget feasible}, if~~~$\sum_{i\in A} p_i(b) \leq B$ for every $b$.
\end{enumerate}
\end{definition}

When referring to randomized mechanisms, the notion of truthfulness we use is \textit{universal truthfulness}, which means that the mechanism is a probability distribution over deterministic truthful mechanisms. 

For single-parameter problems we use the characterization by Myerson \cite{Myerson81} for deriving truthful mechanisms. In particular, we say that an outcome rule $f$ is {\em monotone}, if for every agent $i\in A$, and any vector of cost declarations $b$, if $i\in f(b)$, then $i\in f(b_i', b_{-i})$ for $b_i' \leq b_i$. This simply means that if an agent is selected in the outcome by declaring a cost $b_i$, then he should also be selected if he declares a lower cost.

\begin{lemma}\label{lem:myerson}
Given a monotone algorithm $f$, there is a unique payment scheme $p$ such that $(f, p)$ is a truthful and individually rational  mechanism, given by
\begin{displaymath}
 p_i(b)= \left\{ \begin{array}{ll}
\sup_{b_i\in [c_i, \infty)} \{b_i: i\in f(b_i, b_{-i})\}\,, & \textrm{\emph{ if\ \  }} i\in f(b)\\
0\,, & \textrm{\emph{ otherwise}}
\end{array} \right.
\end{displaymath}
\end{lemma}

Lemma \ref{lem:myerson} is known as Myerson's lemma, and the payments are often referred to as {\it threshold payments}, since they indicate the threshold at which an agent stops being selected. Myerson's lemma simplifies the design of truthful mechanisms by focusing only on constructing monotone algorithms and not having to worry about the payment scheme. Nevertheless, in the setting we study here budget feasibility clearly complicates things further. For all the algorithms presented in the next sections, we always assume that the underlying payment scheme is given by Myerson's lemma.

\section{Deterministic Mechanisms for Submodular Objectives}
\label{sec:submod}

We begin our exposition with submodular valuations, and show in Section \ref{sub:coverage} how to obtain an improved approximation for a subclass of such functions. To do this, we exploit the approach by Chen et al.~\cite{ChenGL11}, starting with their mechanism, shown below:\medskip

\begin{algorithm}[H]
	\DontPrintSemicolon 
	\NoCaptionOfAlgo
	\SetAlgoRefName{\textsc{Mechanism-SM}}
	Set $A=\{i\ |\ c_i\le B\}$ and $i^*\in\argmax_{i\in A}v(i)$ \;
	\vspace{2pt} \If{$\frac{1+4e+\sqrt{1+24e^2}}{2(e-1)}\cdot v(i^*) \ge \opt(A\mysetminus\{i^*\},  v, c_{-i^*}, B)$}{\vspace{2pt}\Return $i^*$}
	\Else{\Return \textsc{Greedy-SM}$(A, v, c, B/2)$}
	\caption{\textsc{Mechanism-SM}$(A, v, c, B)$ \cite{ChenGL11}} \label{fig:alg-1} 
\end{algorithm}\medskip 

In \ref{fig:alg-1}, an agent $i^*$ of maximum value is compared with 
an optimal solution at the instance $A\mysetminus\{i^*\}$ with budget $B$. Then, either $i^*$ or \textsc{Greedy-SM}$(A,  v, c, B/2)$ is returned. \ref{fig:alg-2} is a greedy algorithm that picks agents according to their ratio of marginal value over cost, given that this cost is not too large. For the sake of presentation, we assume the agents are sorted in descending order with respect to this ratio. The marginal value of each agent is calculated with respect to the previous agents in the ordering, i.e.~$1=\argmax_{j\in A}\frac{v(j)}{c_j}$ and $i=\argmax_{j\in A\setminus [i-1]}\frac{v([j])-v([j-1])}{c_j}$ for $i\ge 2$.\medskip

\begin{algorithm}[H]
	\DontPrintSemicolon 
	\NoCaptionOfAlgo
	\SetAlgoRefName{\textsc{Greedy-SM}}
	Let $k=1$ and $S=\emptyset$ \;
	\While{$k\le |A|$ {\rm{\textbf {and}}} $c_k \le \frac{B}{2}\cdot \frac{v(S\cup\{k\})-v(S)}{v(S\cup\{k\})}$}{
		$S=S\cup\{k\}$ \;
		$k=k+1$
	}
	\Return $S$
	\caption{\textsc{Greedy-SM}$(A, v, c, B/2)$ \cite{ChenGL11}} \label{fig:alg-2} 
\end{algorithm}\medskip 

For simplicity, we drop the valuation function and the cost vector from the arguments, e.g., we write \ref{fig:alg-2}$(A, B/2)$ instead of \ref{fig:alg-2}$(A, v, c, B/2)$. From the work of Chen et al.~\cite{ChenGL11} the following lemma is inferred.

\begin{lemma}\label{lem:GreedySM}
	\ref{fig:alg-2}$(A, B/2)$ is monotone and outputs a set $S$ such that $v(S)\ge \frac{e-1}{3e} \cdot \opt(A, B) - \frac{2}{3}\cdot v(i^*)$. Using the threshold payments of Myerson's lemma, the mechanism is truthful, individually rational, and budget feasible.
\end{lemma}

\ref{fig:alg-1} is deterministic and by using Lemma \ref{lem:GreedySM}, it can be shown that it achieves an approximation factor of $8.34$ for any nondecreasing submodular objective. However, it is not guaranteed to run in polynomial time, since we need to compute $\opt(A\mysetminus\{i^*\}, B)$, and more often than not, submodular maximization problems turn out to be $\NP$-hard. An obvious question here is whether we can use an approximate solution instead, but it is not hard to see that by doing so we might sacrifice truthfulness. As a way out, Chen et al.~\cite{ChenGL11} mention that instead of $\opt(A\mysetminus\{i^*\}, B)$, an optimal solution to a fractional relaxation of the problem can be used.\footnote{Intuitively, this would maintain truthfulness because no losing agent can force the mechanism to run \ref{fig:alg-2} without lowering his bid below his current cost. This is due to the fact that $\opt_f$ is nonincreasing with respect to the bid of each agent.} Although this does not always make the mechanism run in polynomial time, it helps in some cases.  

Suppose that for a specific submodular objective, the budgeted maximization problem can be expressed as an ILP, the corresponding LP relaxation of which can be solved in polynomial time. Further, suppose that for any instance $I$ and any budget $B$, the optimal fractional solution $\opt_f(I, B)$ is within a constant factor of the optimal integral solution $\opt(I, B)$. Then replacing $\opt(A\mysetminus\{i^*\}, B)$ by $\opt_f(A\mysetminus\{i^*\}, B)$ in \ref{fig:alg-1} still gives a truthful, constant approximation. In fact, we give a  variant of \ref{fig:alg-1} below, where the constants have been appropriately tuned, so as to optimize the achieved approximation ratio.

Specifically, suppose that the valuation function is such that $\opt_f(I, B) \le \rho\cdot \opt(I, B)$, for any $I$ and any $B$. Let $\gamma=\sqrt{1+4(\rho-1)e+4(\rho^2+4\rho+1)e^2}$ and $\alpha = \frac{1+2(\rho +1)e+\gamma}{2(e-1)}$.\medskip

\begin{algorithm}[H]
	\DontPrintSemicolon 
	\NoCaptionOfAlgo
	\SetAlgoRefName{\textsc{Mechanism-SM-frac}}
	Set $A=\{i\ |\ c_i\le B\}$ and $i^*\in\argmax_{i\in A}v(i)$  \;
	\vspace{2pt} \If{$\alpha \cdot v(i^*) \ge \opt_f(A\mysetminus\{i^*\},  v, c_{-i^*}, B)$\label{line:frac2}} {\vspace{2pt}\Return $i^*$}
	\Else{\Return \textsc{Greedy-SM}$(A, v, c, B/2)$}
	\caption{\textsc{Mechanism-SM-frac}$(A, v, c, B)$} \label{fig:alg-3} 
\end{algorithm} 

\begin{theorem}\label{thm:mechSMfrac}
	\ref{fig:alg-3} is truthful, individually rational, and bud\-get feasible with approximation ratio $\frac{2(\rho +2)e - 1 + \gamma}{2(e-1)}$, where $\rho$ and $\gamma$ are as above. Moreover, it is deterministic and runs in polynomial time given a polynomial time exact algorithm for computing $\opt_f(A\mysetminus\{i^*\}, B)$.
\end{theorem}

The proof of Theorem \ref{thm:mechSMfrac} follows closely the corresponding result for \ref{fig:alg-1} \cite{ChenGL11}, but we include it here for completeness. 

\begin{proof}
For truthfulness and individual rationality, it suffices to show that  the allocation rule is monotone, i.e.~a winning agent $j$ remains a winner if he decreases his cost to $c_j' < c_j$. If $j=i^*$ then clearly his bid is irrelevant and he remains a winner. If $j\neq i^*$ and he was a winner, then by reducing the cost to $c_j'$, the mechanism will still execute \ref{fig:alg-2}, because $\opt_f(A\mysetminus\{i^*\}, B)$ is higher than before. Hence $j$ remains a winner due to the monotonicity of \ref{fig:alg-2} (Lemma \ref{lem:GreedySM}). This argument also highlights why we cannot in general use an arbitrary approximation algorithm instead of $\opt_f(A\mysetminus\{i^*\}, B)$, since we cannot predict how the solution  is affected when the cost changes from $c_j$ to $c_j'$.

Regarding budget feasibility, under the threshold payment scheme of Lemma \ref{lem:myerson}, we either have to pay agent $i^*$ the whole budget, or pay the winners of \ref{fig:alg-2}$(A, B/2)$ the maximum bid that guarantees them to win in \ref{fig:alg-1}. We stress that in the latter case, the payments are upper bounded by the payments induced by running \ref{fig:alg-2}$(A, B/2)$ alone. This holds because $\opt_f(A\mysetminus\{i^*\}, B)$ is decreasing in the cost of each agent, and so line \ref{line:frac2} imposes an extra upper bound on the cost that each agent can report and still be a winner. Hence, budget feasibility follows from the budget feasibility of \ref{fig:alg-2}$(A, B/2)$ (Lemma \ref{lem:GreedySM}).

Finally, for the approximation ratio we consider two cases. If the mechanism returns $i^*$, then 
\[ \alpha \cdot v(i^*) \ge \opt_f(A\mysetminus\{i^*\}, B) \ge \opt(A\mysetminus\{i^*\}, B) \ge \opt(A, B) - v(i^*) \,,\]
and therefore $\opt(A, B)\le (\alpha+1)\cdot v(i^*) = \frac{2(\rho +2)e - 1 + \gamma}{2(e-1)}\cdot v(i^*)$.

On the other hand, if \ref{fig:alg-2}$(A, B/2)$ is executed, and $S$ is the set of agents returned, then
\begin{equation}
\alpha \cdot v(i^*) < \opt_f(A\mysetminus\{i^*\}, B) \le \rho\cdot \opt(A\mysetminus\{i^*\}, B) \le \rho\cdot \opt(A, B)\,.  \label{eq:i*small} 
\end{equation} 
Combining \eqref{eq:i*small} with the approximation from Lemma \ref{lem:GreedySM} we have 
\[\opt(A, B)\le \frac{e}{e-1}\left( 3v(S)+2v(i^*)\right) < \frac{e}{e-1}\left( 3v(S)+\frac{2\rho}{\alpha}\opt(A, B)\right)\,, \]
and therefore $\opt(A, B)\le \frac{3\alpha e}{\alpha (e-1) - 2\rho e}\cdot v(S) = \frac{2(\rho +2)e - 1 + \gamma}{2(e-1)}\cdot v(S)$, where the last equality is just a matter of calculations.

Therefore, in both cases, the output of the mechanism is a $\frac{2(\rho +2)e - 1 + \gamma}{2(e-1)}$-approximation of \linebreak $\opt(A, B)$.
\end{proof}

\subsection{Budgeted Max Weighted Coverage}
\label{sub:coverage}

We consider the class of {\it weighted coverage valuations}, a special class of submodular functions. 
Their unweighted version was studied by Singer in \cite{Singer10} and \cite{Singer12}, motivated by the problem of influence maximization over social networks. Imagine a company that tries to promote a new product and as part of its marketing campaign decides to advertise (or even sell at a promotional price) the product to selected influential nodes. Suppose that each node $i$, is able to influence some set of other nodes, but this comes at a cost $c_i$ (cost of advertising and convincing $i$). Then, if there is a budget available for the campaign, the goal would be to select a set of initial nodes respecting the budget, so as to maximize the (weighted) union of people who are eventually influenced. This gives rise to the following problem.\medskip

\noindent\emph{Budgeted Max Weighted Coverage.} Given a set of subsets $\{S_i\ |\ i\in [m]\}$ of a ground set $[n]$, along with costs $c_1, c_2,...,c_m$, on the subsets, weights $w_1, ..., w_n$, on the ground elements, and a positive budget $B$, find $X\subseteq [m]$ so that $v(X)=\sum_{j\in \bigcup_{i\in X}S_i} w_j$ is maximized subject to $\sum_{i\in X} c_i \le B$.\medskip

In the definition above, $S_i$ is the set of people that agent $i$ can influence. On a different note, the problem can also be thought of as a crowd-sourcing problem, where each (single-minded) worker $i$ is able to execute only the set of tasks $S_i$. 

In \cite{Singer12}, Singer takes an approach similar to what led to \ref{fig:alg-3}, but suggests a different polynomial time mechanism for Budgeted Max Coverage that is deterministic, truthful, budget feasible, and achieves approximation ratio $31.03$. 
Here we generalize and improve this result by showing that there is a deterministic, truthful, budget feasible, polynomial time $15.45$-approximate mechanism for the Budgeted \emph{Weighted} Max Coverage problem. 


For all $j\in[n]$ define $T_j=\{i\ |\ j\in S_i\}$. We begin with a LP formulation of this problem,
where without loss of generality we assume that $c_i\le B, \forall i\in[n]$ (otherwise we could just discard any subsets with cost greater than $B$).
\begin{IEEEeqnarray}{rCll}
	\mbox{maximize:\ \ } & & \sum_{j\in [n]} w_j z_j & \label{eq:bwmc1}\\
	\mbox{subject to:\ \ } & & \sum_{i\in T_j} x_i \ge z_j\ ,\qquad & \forall j\in[n]\\
	& & \sum_{i\in [m]} c_i x_i \le B &  \\
	& & 0\le x_i, z_j \le 1\ ,& \forall i\in [m],\, \forall j\in [n] \label{eq:bwmc4}\\
	& & x_i\in \{0,1\}\ ,& \forall i\in [m] \label{eq:bwmc5}
\end{IEEEeqnarray}

It is not hard to see that \eqref{eq:bwmc1}-\eqref{eq:bwmc5} is a natural ILP formulation for Budgeted Max Weighted Coverage and \eqref{eq:bwmc1}-\eqref{eq:bwmc4} is its linear relaxation. For the rest of this subsection, let $\opt(I, B)$ and $\opt_f(I, B)$ denote the optimal solutions to \eqref{eq:bwmc1}-\eqref{eq:bwmc5} and \eqref{eq:bwmc1}-\eqref{eq:bwmc4} respectively for instance $I$ and budget $B$. 

To show how these two are related we will use the technique of pipage rounding \cite{AgeevS99,AgeevS04}. Although we do not provide a description of the general pipage rounding technique, the proof of Lemma \ref{lem:bwmc} below is self-contained. 
We should note here that Ageev and Sviridenko \cite{AgeevS04} use the above linear programs (as well as the nonlinear program in the proof of Lemma \ref{lem:bwmc}) to obtain an $\frac{e}{e-1}$-approximation LP-based algorithm that uses pipage rounding on a number of different instances of the problem.\footnote{In fact, Ageev and Sviridenko \cite{AgeevS04} study the hitting set version of this problem, but both problems have essentially the same linear program formulation.} However, in their algorithm $\opt(I, B)$ is never compared directly to $\opt_f(I, B)$, and therefore we cannot get the desired bound from there. 

\begin{lemma}\label{lem:bwmc}
	Given the fractional relaxation \eqref{eq:bwmc1}-\eqref{eq:bwmc4} for Budgeted Max Weighted Coverage, we have that for any instance $I$ and any budget $B$
	\[\opt_f(I, B) \le \frac{2e}{e-1}\cdot \opt(I, B) \,.\] 
\end{lemma}

\begin{proof}
	Given any feasible solution $x, z$ to \eqref{eq:bwmc1}-\eqref{eq:bwmc4}, the value of \eqref{eq:bwmc1} is upper bounded by $L(x)=  \sum_{j\in [n]} w_j\cdot \min\{1, \sum_{i\in T_j} x_i \}$, since $z_j \le  \min\{1, \allowbreak \sum_{i\in T_j} x_i\}$ for any $j\in [n]$. In particular, if $x^*, z^*$ is an optimal (fractional) solution to \eqref{eq:bwmc1}-\eqref{eq:bwmc4}, then the value of \eqref{eq:bwmc1} is exactly $L(x^*) = \opt_f(I, B)$.
	
	Next we  consider the nonlinear program 
	\begin{IEEEeqnarray}{rCl}
		\mbox{maximize:\ \ } & & F(x) = \sum_{j\in [n]} w_j \bigg(1- \prod_{i\in T_j}(1-x_i)\bigg)   \label{eq:nlbwmc1}\\
		\mbox{subject to:\ \ }& & \sum_{i\in [m]} c_i x_i \le B \label{eq:nlbwmc2} \\
		& & 0\le x_i \le 1\ , \forall i\in [m] \label{eq:nlbwmc3}
	\end{IEEEeqnarray}
	and we observe that $F(x)\ge (1-1/e) L(x)$ for any feasible vector $x$. This follows from the fact that $( 1-( 1-1/k) ^k )\ge (1-1/e)$ for any $k\ge 1$, and the following inequality, derived in Goemans and Williamson \cite{GoemansW94} (Lemma 3.1 in their work):
	\[1- \prod_{i\in [k]}(1-y_i) \ge ( 1-( 1-1/k) ^k ) \min\bigg\{1, \sum_{i\in [k]} y_i\bigg\} \,.\]
	So, if $x^*, z^*$ is an optimal solution to \eqref{eq:bwmc1}-\eqref{eq:bwmc4} we have
	\[F(x^*)\ge (1-1/e) L(x^*) = (1-1/e) \opt_f(I, B)\,.\]
	
	However, $x^*$ may have several fractional coordinates. Our next step is to transform $x^*$ to a vector $x'$ that has at most one fractional coordinate and at the same time $F(x')\ge F(x^*)$. To this end, we show how to reduce the fractional coordinates by (at least) one in any feasible vector with at least two such coordinates.
	
	Consider a feasible vector $x$, and suppose $x_i$ and $x_j$ are two non integral coordinates. Let $x^{i,j}_{\varepsilon}$ be the vector we get if we replace $x_i$ by $x_i+\varepsilon$ and $x_j$ by  $x_j - \varepsilon c_i / c_j$ and leave every other coordinate of $x$ the same. Note that the function $\bar{F}(\varepsilon)=F(x^{i,j}_{\varepsilon})$, with respect to $\varepsilon$, is either linear or a polynomial of degree 2 with positive leading coefficient. That is, $\bar{F}(\varepsilon)$ is convex. 
	
	Notice now that $x^{i,j}_{\varepsilon}$ always satisfies the budget constraint \eqref{eq:nlbwmc2}, and also satisfies \eqref{eq:nlbwmc3} as long as 
	$\varepsilon\in \left[ \max\left\lbrace -x_i, (x_j - 1)c_j / c_i\right\rbrace  , \min\left\lbrace 1-x_i, x_j c_j / c_i\right\rbrace \right]$.
	Due to convexity, $\bar{F}(\varepsilon)$ attains a maximum on one of the endpoints of this interval, say at $\varepsilon^*$. Moreover, at either endpoint at least one of $x_i+\varepsilon^*$ and $x_j - \varepsilon^* c_i/c_j$ is integral. That is, $x^{i,j}_{\varepsilon^*}$ has at least one more integral coordinate than $x$ and $F(x^{i,j}_{\varepsilon^*})\ge F(x)$.
	
	Hence, initially $x\leftarrow x^*$. As long as there exist two non integral coordinates $x_i$ and $x_j$ we set $x\leftarrow x^{i,j}_{\varepsilon^*}$ as described above. This procedure runs for at most $n-1$ iterations, and outputs a feasible vector $x'$ that is integral or almost integral and $F(x')\ge F(x^*)$.
	
At this point we should note that when all the $x_i$s are integral then the objectives \eqref{eq:bwmc1} and \eqref{eq:nlbwmc1} have the same value if we set $z_j = \min\{1, \sum_{i\in T_j} x_i\}$ for all $j\in [n]$. Specifically, if $x$ is any feasible integral vector, we have $F(x)\le \opt(I, B)$. So, if $x'$ is integral then
	\[\opt_f(I, B) = L(x^*) \le \frac{e}{e-1} \cdot F(x^*) \le \frac{e}{e-1} \cdot F(x') \le \frac{e}{e-1} \cdot \opt(I, B)\,, \]
	and we are done. Thus, suppose that $x'$ has exactly one fractional coordinate, say $x_r$. Let $x^{i}$ be the vector we get if we set $x_r$ to $i\in\{0,1\}$ and leave every other coordinate of $x'$ the same. Note that $x^{1}-x^{0}$ corresponds to the vector that has $1$ in the $r$th coordinate and  $0$ everywhere else. Clearly, $\opt(I, B)\ge \max\left\lbrace F(x^{0}), F(x^{1}-x^{0})\right\rbrace $ (where $x^{1}-x^{0}$ is feasible since we have discarded subsets with $c_i > B$). Moreover, $F$ on integral vectors is submodular, and hence subadditive, therefore $F(x^{1}) - F(x^{0})\le F(x^{1}-x^{0})$. Finally, it is easy to see that $F$ is increasing with respect to any single coordinate, so $F(x^{1}) \ge F(x')$. Combining all the above we get
	\begin{IEEEeqnarray*}{rCl}
		\opt_f(I, B) & = & L(x^*) \le \frac{e}{e-1} \cdot F(x^*) \le \frac{e}{e-1} \cdot F(x')\le \frac{e}{e-1} \cdot F(x^{1})\\
		& \le & \frac{e}{e-1} \cdot  \left( F(x^{0}) + F(x^{1}-x^{0}) \right) \le \frac{2e}{e-1} \cdot  \opt(I, B) \,,
	\end{IEEEeqnarray*}
	thus completing the proof.
\end{proof}

Combining Theorem \ref{thm:mechSMfrac} and Lemma \ref{lem:bwmc} we get the following result.

\begin{corollary}
	There exists a deterministic, truthful, individually rational, budget feasible  $15.45$-appro\-ximate mechanism for Budgeted Max Weighted Coverage that runs in polynomial time. 
\end{corollary}

\section{Going Beyond Submodularity}
\label{sec:nonsubmod}

Going beyond submodular valuations seems even more challenging. The first attempt to design a truthful mechanism for a budgeted maximization problem with a non-submodular objective was due to Chen et al.~\cite{ChenGL11} who gave a $(2+\sqrt{2})$-approximation mechanism for a non-submodular variation of Knapsack.
For the more general class of  subadditive functions Dobzinski et al.~\cite{DobzinskiPS11} suggested a randomized $O(\log^2 n)$-approximation mechanism, and later, Bei et al.~\cite{BeiCGL12} provided randomized, truthful, budget feasible mechanisms with approximation ratio $768$ for XOS objectives and $O\big(\frac{\log{n}}{\log{\log{n}}}\big)$ for subadditive objectives. 

More recently, Goel et al.~\cite{GoelNS14} study a budgeted maximization problem with matching constraints, which is not submodular, and they achieve an approximation ratio of $3+o(1)$ with a deterministic mechanism, but under the large market assumption\footnote{A market is said to be large if the number of participants is large enough that no single person can affect significantly the market outcome, i.e.~$\max_i c_i / B=o(1)$. }
(their mechanism has an unbounded ratio in general).
Essentially, they use the same greedy approach with Singer \cite{Singer10} and Chen et al.~\cite{ChenGL11} but seen as a descending price auction.
A very similar mechanism was also briefly discussed in Anari et al.~\cite{AnariGN14} for Knapsack under the large market assumption. 

We are building on this idea of gradually decreasing a global upper bound on the payment per value ratio to get all the results of this section. We first use Budgeted Max Weighted Matching in Subsection \ref{sub:matching}, as an illustrative example of how this approach works, but the exact same approach gives the same approximation guarantees for a number of different XOS problems that can be seen as appropriately restricted generalizations of Knapsack. We elaborate further on this in Subsection \ref{sub:other}, and we even extend these ideas to problems where the unbudgeted versions are not easy.

\subsection{Budgeted Max Weighted Matching}
\label{sub:matching}
We revisit the following budgeted matching problem.\medskip 

\noindent\emph{Budgeted Max Weighted Matching.} Given a budget $B$, and a graph $G = (V, E)$, where each edge $e_i\in E$ has a cost $c_i$ and a value $v_i$, find a matching $M$ of maximum value subject to $\sum_{i\in M} c_i \le B$.\medskip

Here we study the mechanism design version of the problem, where the values are known to the mechanism and the edges are viewed as single-parameter strategic agents whose cost is private information.\footnote{The work of Singer~\cite{Singer10} also studies a type of a budgeted matching problem. That objective, however,  is OXS (a subclass of submodular objectives), and differs significantly from ours, which is not submodular~\cite{Singer-personal}.} 
Note that in order to formulate the problem to fit the general description given in the beginning of Section \ref{sec:defs}, we can define the valuation function as follows (as also mentioned in \cite{BeiCGL12}): 
for any subset of edges $S\subseteq E$, $v(S)$ is taken to be the value of the maximum weighted matching of $G$ that only uses edges in $S$. This function turns out to be XOS, but not submodular (see the next proposition). Hence, by \cite{BeiCGL12}, there exists a randomized, 768-approximation, that is truthful and budget feasible. 

\begin{proposition}
The function $v(\cdot)$ defined above is XOS, but not submodular.
\end{proposition}

\begin{proof}
To prove that $v(\cdot)$ is not submodular, we may consider the following example:
\begin{center}
{\scalebox{0.34} {\includegraphics{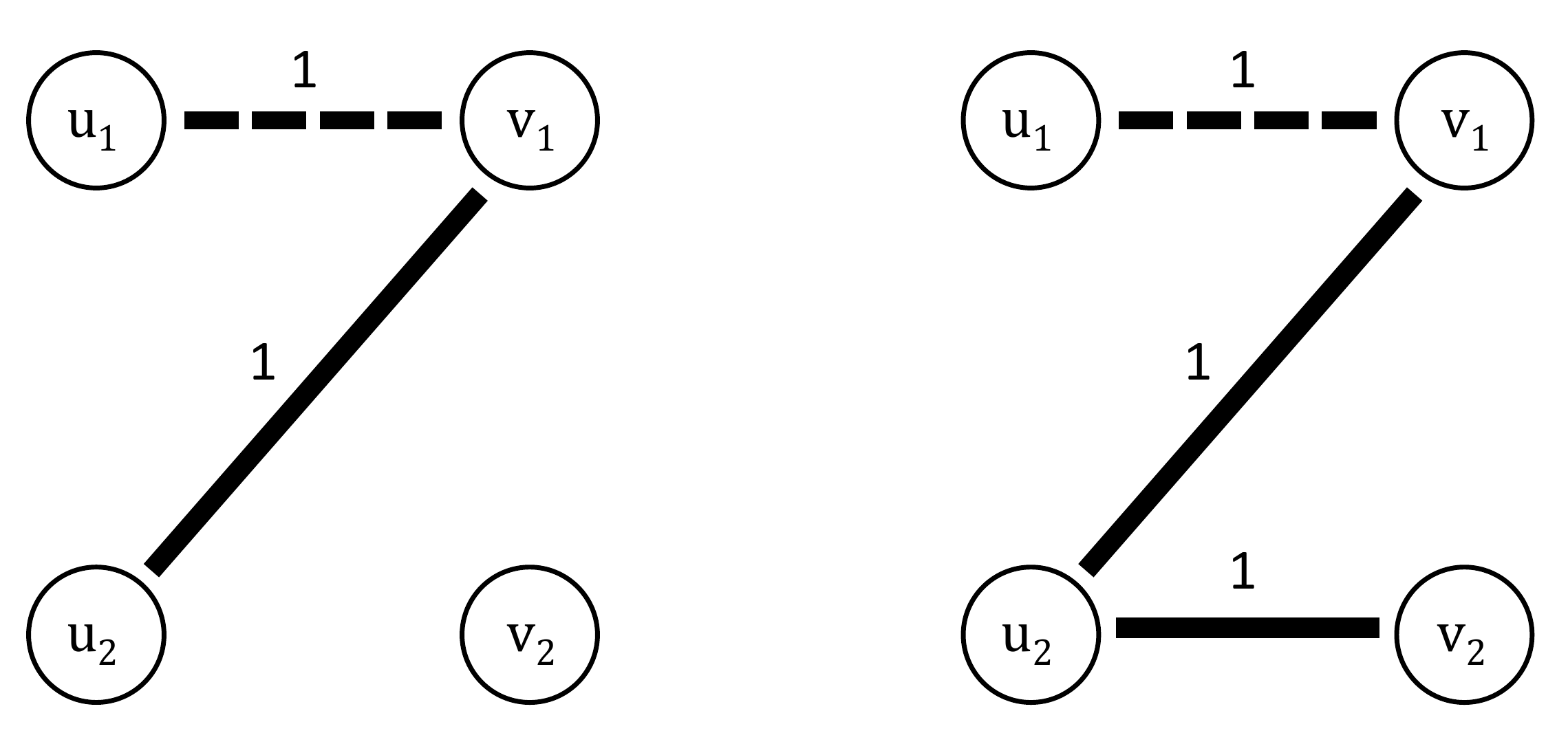}}}
\end{center}

Let $A=\{u_2v_1\}$, $B= \{u_2v_1, u_2v_2\}$ and add the dashed edge $u_1v_1$ to both sets. Then we have that $v(A \cup \{u_1v_1\}) -v(A) =1-1=0 < 1=2-1 =v(B\cup \{u_1v_1\}) -v(B)$.

Now to prove that $v(\cdot)$ is XOS, let $\{M_1, M_2, ...,M_r\}$ be the finite set of all possible matchings of a given graph $G$, and set $\alpha_j (S)= \sum_{i\in S\cap M_j} v_i \,,$ for $ j \in \{1,...,r\}$, $ S \subseteq E(G)$. Note that each $\alpha_j$ is an additive function. Indeed, $\alpha_j (S)= \sum_{i\in S} v_{j, i}\,,$ where $v_{j, i}$ is equal to $v_i$ if $i\in  M_j$ and $0$ otherwise.
Since $v(S)$ is defined to be the value of the maximum weighted matching of $G$ that only uses edges in $S$, we have that $v(S) = \max \{\alpha_1(S), \alpha_2(S), ...,\alpha_r(S) \}$. 
\end{proof}

We provide both deterministic and randomized polynomial time mechanisms with a much improved approximation ratio,  based on selecting an outcome among two candidate solutions. The first solution comes from the greedy mechanism \ref{fig:alg-4} described below. The main idea behind the mechanism is that in each iteration there is an implicit common upper bound on the rate that determines the payment of each winner in the candidate outcome of that iteration. More specifically, if the $i$th iteration is the final iteration (i.e.~the condition in line \ref{line:test} is true), the common payment per value for each of the winners is upper bounded by $\min\{B/v(M), c_{i-1}/v_{i-1}\}$. This upper bound decreases with each iteration, while the set of active agents is shrinking, until budget feasibility is achieved. At the same time we ensure the mechanism is monotone and returns enough value.

We assume that the mechanism also takes as input a deterministic exact algorithm $f$ for the unbudgeted Max Weighted Matching, e.g., Edmond's algorithm \cite{Edmonds1965b}.  Later, in Subsection \ref{sub:other} the choice of  $f$ will depend on the underlying unbudgeted problem. Finally, note that  our mechanisms are named after the generalization we study in Subsection \ref{sub:other}, namely Independence System Knapsack problems.\medskip

\begin{algorithm}[H]
	\DontPrintSemicolon 
	\NoCaptionOfAlgo
	\SetAlgoRefName{\textsc{Greedy-ISK}}
	Set $A=\{i\ | \ c_i\le B\}$  \;
	\vspace{2pt} Possibly rename elements of $A$ so that $\frac{c_1}{v_1} \geq \frac{c_2}{v_2} \geq ... \geq \frac{c_m}{v_m} $ \;
	\For{ $i=1$ to $m$}{$M = f(A, v)$ \label{line:f} \;
		\vspace{2pt} \If{$v(M) \cdot \frac{c_i}{v_i} \leq B$ \label{line:test}}{\vspace{2pt} \Return $M$  \;}
		\Else{$A=A\mysetminus \{i\}$}}
	\caption{\textsc{Greedy-ISK}$(A, v, c, B, f )$} \label{fig:alg-4} 
\end{algorithm}\medskip 


\noindent We now exhibit some desirable properties of \ref{fig:alg-4}, starting with truthfulness.

\begin{lemma} \label{lem:gutru} 
	Mechanism \ref{fig:alg-4} is monotone, and hence truthful and individually rational. 
\end{lemma}

\begin{proof} 
	By Lemma \ref{lem:myerson}, we just need to show that the allocation rule is monotone, i.e.~a winning agent remains a winner if he decreases his cost.
	Initially note that in line \ref{line:f} the mechanism computes an optimal matching $M$ (without a budget constraint) using only the values of the edges, thus it cannot be manipulated given the set of active edges $A$. 
	
	Fix a vector $c_{-j}$ for the costs of the other agents, and suppose that when agent $j$ declares $c_j$, he is in the matching $M$ returned in the final iteration, say $k$, of \ref{fig:alg-4}. Let agent $j$ now report $c_j' < c_j$ to the mechanism. This makes him agent $j' \geq j$ in the new instance, but does not affect the relative ordering of the other agents (although a few of them may move down one position). Therefore, \ref{fig:alg-4} will run exactly as before for each iteration $i<k$ and in the beginning of the $k$th iteration, it will produce the exact same matching $M$. Then in line \ref{line:test}, there are 2 cases to examine. If in the initial instance $j>k$, then we have the exact same ratio $\frac{c_k}{v_k}$ to consider, and the algorithm will terminate with $M$ (since it did so in the initial instance). In the second case, $j=k$ in the initial instance. This means that now at the $k$th iteration, we either have the same agent with the reduced ratio $\frac{c_k'}{v_{k}}$ (since now $c_k'=c_j'$) or we have the agent who was in position $k+1$ in the initial instance with ratio equal to the original $\frac{c_{k+1}}{v_{k+1}}$. Therefore, the new ratio $\frac{c_{k}}{v_{k}}$ that the algorithm considers in this iteration is at most equal to the original ratio $\frac{c_{k}}{v_{k}}$. Thus, the condition in line \ref{line:test} is satisfied, and the mechanism will return $M$. We conclude that an agent who is in the matching, remains in the matching by decreasing his cost.
\end{proof}

We also make the following remark, which can be derived by the same arguments used in the proof of Lemma \ref{lem:gutru}. This property is crucial for derandomizing our mechanisms both here and in the next subsection.

\begin{rem}\label{rem:guprop} 
There is no agent $i$ that can manipulate the output set of \ref{fig:alg-4} given that $i$ is guaranteed to be a winner. That is, fix $c_{-i}$ and let $M$ and $M'$ be the winning sets when $i$ bids $c_i$ and $c_i'$ respectively; if  $i \in M \cap M'$,   then $M=M'$.
\end{rem}

We move on to prove that the mechanism will never exceed the budget $B$, by establishing an appropriate upper bound on every winning bid.

\begin{lemma} \label{lem:gubf} 
	Mechanism   \ref{fig:alg-4} is budget feasible.
\end{lemma}

\begin{proof} 
We will show that the threshold payment of Lemma \ref{lem:myerson}  cannot be higher than $\frac{v_iB}{v(M)}$ for any winning agent $i$. Fix a  vector $c_{-i}$ for all agents other than $i$ and recall that the threshold payment, given $c_{-i}$, is the maximum cost that $i$ can declare and still be included in the solution. So, towards a contradiction, suppose that agent $i$ declares a cost $c_i >\frac{v_iB}{v(M)}$ and he is a winner. Let $j$ denote the iteration where the mechanism \ref{fig:alg-4} terminates and the matching $M$ is returned.  
	By the construction of the mechanism, and since $i\in M$, we have that $\frac{c_j}{v_j}\geq \frac{c_i}{v_i}$. Since $j$ is the last iteration, we also have by line \ref{line:test} that $v(M)\frac{c_j}{v_j} \leq B$. Hence  $v(M)\frac{c_i}{v_i} \leq v(M)\frac{c_j}{v_j} \leq B$ that leads to the contradiction $c_i \leq \frac{v_iB}{v(M)}$. Therefore, the payment of each winning agent $i$ is bounded by  $\frac{v_iB}{v(M)}$, and the total payment of the mechanism is  $\sum_{i\in M} p_i \leq \sum_{i\in M}\frac{v_iB}{v(M)} = B$.
\end{proof}

Finally, we analyze the quality of the solution produced by the greedy mechanism.
\begin{lemma}  \label{lem:guapx}  
	Mechanism \ref{fig:alg-4} produces a matching with value at least $\frac{1}{2}(v(M^*)-v_{i^*})$, where $M^*$ is an optimal solution to the given instance of Budgeted Max Weighted Matching, and $i^*$ has maximum value among the budget feasible edges of $G$, i.e.~$i^*\in\argmax_{i\in F}v(i)$ where $F=\{i\in E(G)\ | \ c_i\le B\}$.
\end{lemma}

\begin{proof}
	Let $M^*$ be an optimal budget feasible matching and $A$ be the set of active edges at the final iteration $j$ of \ref{fig:alg-4} when matching $M$ was returned. We have that $ v(M^*)=v(M^* \cap A)+v(M^* \mysetminus A)$. Since $M^* \cap A$ is a matching with edges from $A$ but $M$ is an optimal such matching, we have that
	\begin{equation}
	\label{eq:f} v(M^* \cap A) \leq v(M)\,.
	\end{equation} 
	In addition, notice that if $i \in M^* \mysetminus A$ then $\frac{c_i}{v_i} \geq \frac{c_{j-1}}{v_{j-1}}$ since $j-1$ is the last edge removed from the set $A$ before the final iteration $j$. Thus,
	\begin{equation}
	\label{eq:s} B \geq \sum_{i\in M^* \mysetminus A}c_i \geq \sum_{i\in M^* \mysetminus A}v_i \cdot\frac{c_{j-1}}{v_{j-1}}\geq v(M^* \mysetminus A) \cdot\frac{c_{j-1}}{v_{j-1}} \,.
	\end{equation}

Now, consider the  $(j-1)$th iteration and call $M'$ the matching produced in that iteration. 
Note that $M' \mysetminus \{j-1\}$ is a  matching containing only edges that are  active during iteration $j$.
 Therefore, $v(M) \geq v(M' \mysetminus \{j-1\})$. Moreover, if $j-1 \in M'$ then $v(M')=v(M' \mysetminus \{j-1\})+v_{j-1}$, while if $j-1 \notin M'$ then $v(M')=v(M' \mysetminus \{j-1\}) \le v(M' \mysetminus \{j-1\})+v_{j-1}$. Using also the fact that $j-1$ was not the final iteration we have
	\begin{equation}
	\label{eq:fo}  \frac{v_{j-1}}{c_{j-1}}\cdot B<v(M') \le v(M' \mysetminus \{j-1\})+v_{j-1} \leq v(M) +v_{i^*}\,.
	\end{equation} 
	By combining  \eqref{eq:s} and  \eqref{eq:fo} we get 
	\begin{equation}
	\label{eq:fi} v(M^* \mysetminus A) \leq v(M)+v_{i^*}\,.
	\end{equation} 
	Finally, combining \eqref{eq:fi} with \eqref{eq:f} we get $ v(M^*)=v(M^* \cap A)+v(M^* \mysetminus A) \leq v(M)+v(M)+v_{i^*}=2v(M)+ v_{i^*}$ and therefore $v(M) \geq \frac{1}{2}(v(M^*)-v_{i^*}) $.
\end{proof}

\noindent We can now state our randomized mechanism  for the problem (where the constants below have been optimized to get the best ratio).\medskip

\begin{algorithm}[H]
	\DontPrintSemicolon 
	\NoCaptionOfAlgo
	\SetAlgoRefName{\textsc{Rand-ISK}}
	Set $A=\{i\ | \ c_i\le B\}$ and $i^*\in\argmax_{i\in A}v(i)$  \;
	\vspace{2pt} With probability 1/3 return $i^*$ and with probability 2/3 return \ref{fig:alg-4}$(A, v, c, B, f )$
	\caption{\textsc{Rand-ISK}} \label{fig:alg-5} 
\end{algorithm} 

\begin{theorem}\label{them:ran}
	\ref{fig:alg-5} is a universally truthful, individually rational, budget feasible, polynomial time randomized mechanism, achieving a 3-approximation in expectation, for the Budgeted Max Weighted Matching problem.
\end{theorem}

\begin{proof}
	Universal truthfulness and individual rationality follow from Lemma \ref{lem:gutru} and the fact that the simple mechanism that returns $i^*$ and pays him $B$ is truthful and individually rational. Regarding budget feasibility, just notice that if $i^*$ is returned then the threshold payment is exactly $B$, otherwise the payments of \ref{fig:alg-4} are used, so budget feasibility follows from Lemma \ref{lem:gubf}. 
	Finally, if $M$ denotes the outcome of \ref{fig:alg-5}, then directly by Lemma \ref{lem:guapx} we have 
	\[ \mathrm E(M)\ge \frac{2}{3} \cdot \frac{1}{2}(v(M^*) - v_{i^*}) +  \frac{1}{3}v_{i^*}  = \frac{1}{3} v(M^*)\,, \]
	thus proving the approximation ratio.
\end{proof}

\noindent {\bf Derandomization.} 
We close this subsection by providing a deterministic polynomial time mechanism with a slightly worse approximation ratio. It is interesting to note that in contrast to \ref{fig:alg-1} or \ref{fig:alg-3}, here ${i^*}$ is directly compared to its alternative, which is just an approximate solution, without sacrificing truthfulness. This is due to Remark \ref{rem:guprop}. Note also that although taking the maximum of two truthful algorithms does not always yield a truthful mechanism, we show this is the case for the mechanism below.\medskip

\begin{algorithm}[H]
	\DontPrintSemicolon 
	\NoCaptionOfAlgo
	\SetAlgoRefName{\textsc{Det-ISK}}
	Set $A=\{i\ | \ c_i\le B\}$ and $i^*\in\argmax_{i\in A}v(i)$ \label{line:det1} \;
	\vspace{2pt} \If{$v_{i^*} \geq \ref{fig:alg-4}(A\mysetminus \{i^*\}, v, c_{-i^*}, B, f )$ \label{line:det2}}{\vspace{2pt}\Return $i^*$ \label{line:det3}}
	\Else{\Return \ref{fig:alg-4}$(A\mysetminus \{i^*\}, v, c_{-i^*}, B, f )$ \label{line:det5}}
	\caption{\textsc{Det-ISK}} \label{fig:alg-6} 
\end{algorithm} 

\begin{theorem}\label{thm:det}
\ref{fig:alg-6} is a truthful, individually rational, budget feasible, polynomial time deterministic mechanism, achieving a 4-approximation ratio for the Budgeted Max Weighted Matching problem. 
\end{theorem}

\begin{proof}
	For truthfulness and individual rationality we show the algorithm is monotone. 
	If $i^*$ wins (lines \ref{line:det2}-\ref{line:det3}) and he decreases his cost, he still wins  since his bid is irrelevant to the outcome. On the other hand, suppose that the mechanism reaches line \ref{line:det5} and let $i \in A\mysetminus \{i^*\}$ be one of the winners. Then, by decreasing his cost, $i$ will remain a winner since the output of \ref{fig:alg-4} will not change (see the proof of Lemma \ref{lem:gutru} and Remark \ref{rem:guprop} after that), and the same branch of the mechanism \ref{fig:alg-6} will be executed again.
	
	Budget feasibility is straightforward and follows from the same arguments used in the proof of Theorem \ref{them:ran}
	
	For the approximation ratio, we begin with some notation. Let $v_G$ be the value of the matching returned by \ref{fig:alg-4}$(A\mysetminus \{i^*\}, v, c_{-i^*}, B, f)$, $M$ be the output of \ref{fig:alg-6}$(A, v, \allowbreak c, B, f)$, $\opt(S)$ be the value of an optimal solution with respect to the set of edges $S\subseteq E(G)$, and finally let $i'$ be an edge of maximum value in the set $A\mysetminus \{i^*\}$. Clearly, if $\opt$ is the value  of an optimal solution to the initial instance, then $\opt=\opt(A)$. Finally, observe that $\opt(A)\leq \opt(A\mysetminus \{i^*\})+v_{i^*}  \leq 2v_G+ v_{i'}+v_{i^*} \leq 2v_G + 2v_{i^*}$. Now if $v_{i^*} \geq v_G$ then $\opt \leq 4v_{i^*} = 4 v(M)$, while if $v_{i^*} < v_G$ then $\opt \leq  4v_G = 4 v(M)$. Thus in any case we have that $\opt\leq 4v(M)$ and this concludes the proof.
\end{proof}

The analysis of \ref{fig:alg-6}  is tight, as shown in the following proposition.

\begin{proposition}
There exist instances where the value of the optimal solution is arbitrarily close to four times the value of the output of \ref{fig:alg-6}.
\end{proposition}

\begin{proof}
We  provide an example where the value achieved by \ref{fig:alg-6}, is almost $1/4$ of the  value of an optimal solution. Consider the following graph $G=(V, E)$, where $v_{1}=v+2\epsilon$, $v_{2}=v$, $v_{3}=v$, $v_{4}=v+\epsilon$, for $\epsilon>0$, $c_{1}= \delta$, $c_{2}=10$, $c_{3}=10$, $c_{4}=\delta$ for $\delta < \frac{5\epsilon}{v}<<10$, and $B=20+2\delta$. 
\begin{center}
{\scalebox{0.38} {\includegraphics{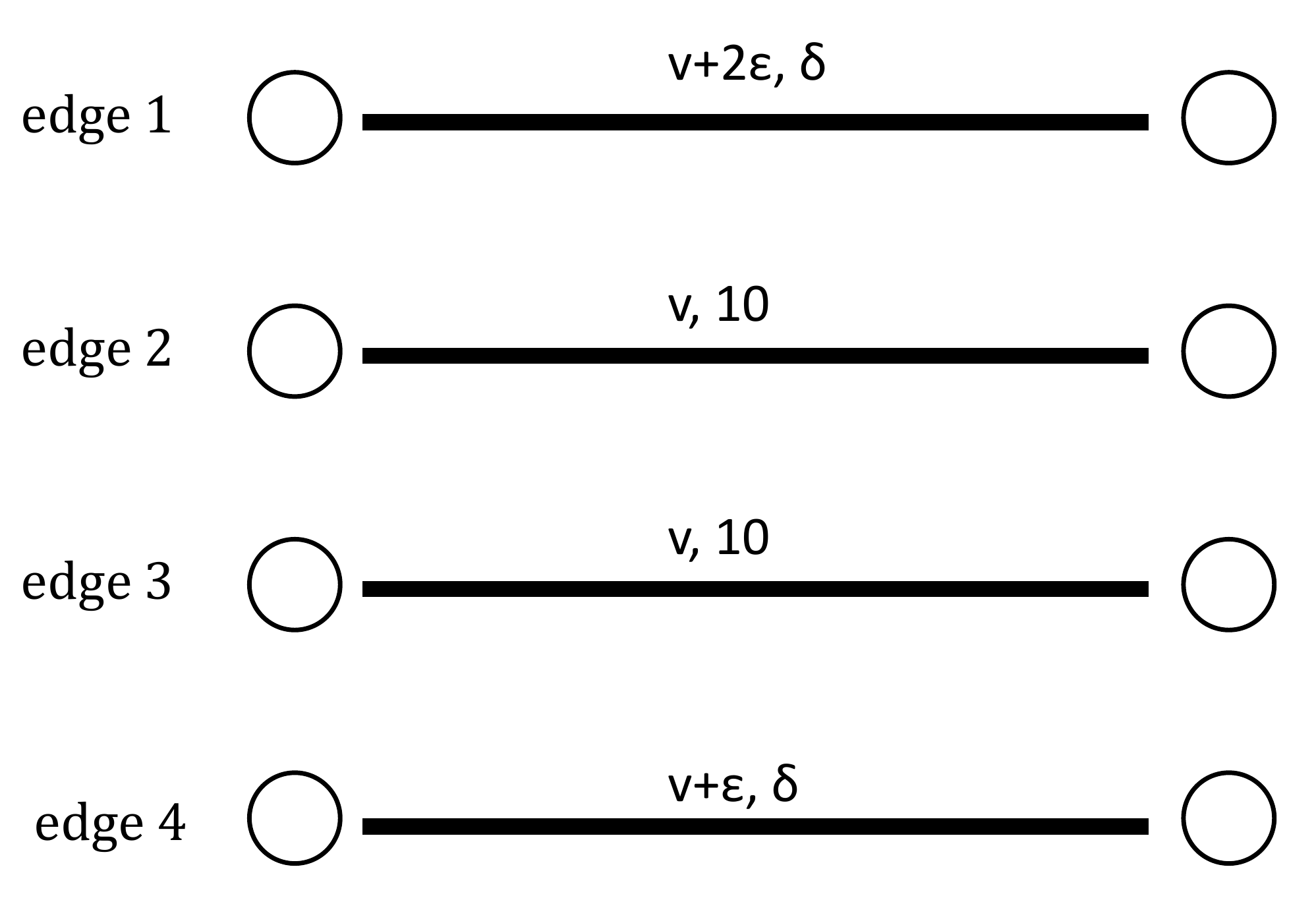}}}
\end{center}
It is easy to check that $E$ is budget feasible, and thus $E$ is the optimal solution with value equal to $ 4v+3\epsilon =\opt$. Now let us examine the value of \ref{fig:alg-6}'s output. The most valuable edge here is edge $1$ with $v_{1}= v +2\epsilon$, so $i^*=1$. On the other hand, \ref{fig:alg-4} orders the remaining edges in the following manner: $\frac{c_{2}}{v_{2}} \geq \frac{c_{3}}{v_{3}} \geq \frac{c_{4}}{v_{4}}$. So by running this instance we have that $(3v+\epsilon)\frac{10}{v}=30+\frac{10\epsilon}{v}>20+\frac{2 \cdot 5\epsilon}{v}>20+2\delta=B$, and thus edge $2$ is excluded. \ref{fig:alg-4} then moves to the next iteration (edges $3$ and $4$ are active), where $(2v+\epsilon)\frac{10}{v}=20+\frac{10\epsilon}{v}=20+\frac{2 \cdot 5\epsilon}{v}>20+2\delta=B$, so edge $3$ is excluded as well. \ref{fig:alg-4} moves to the next iteration (only edge $4$ is active), where $(v+\epsilon)\frac{\delta}{v+\epsilon}=\delta \leq 20+2\delta=B$, so the output is edge $4$ with total value $v+\epsilon$. 

Now we have that $v_{i^*}=v+2\epsilon > v+\epsilon=$  \ref{fig:alg-4}$(U\mysetminus\{i^*\}, B)$ and hence the value of \ref{fig:alg-6}'s output is $v+2\epsilon \simeq \frac{1}{4} (4v+3\epsilon)= \frac{1}{4} \opt$.
\end{proof}

\begin{rem}\label{rem:lower}
	Chen et al.~\cite{ChenGL11} prove lower bounds for Knapsack, namely there is no deterministic (resp. randomized) truthful, budget feasible mechanism for Knapsack that  achieves an approximation ratio better than $1+\sqrt{2}$ (resp. $2$). These lower bounds hold here as well, because when the given graph $G$ is a matching to begin with, Budgeted Max Weighted Matching reduces to Knapsack.
\end{rem}

\subsection{A Generalization to Other Objectives}
\label{sub:other}

Our approach can tackle a number of different problems that have certain structural similarities with Budgeted Max Weighted Matching. Here, we define a class of such problems for which \ref{fig:alg-4}---given an appropriate subroutine $f$---produces truthful, individually rational, budget feasible mechanisms with good approximation guarantees. 

Two crucial properties of the matching problem were used in the previous subsection: (i) every subset of a matching is itself a matching, and (ii) the objective function becomes additive when restricted to matchings. These two properties is all we need, and note that (i) and (ii) are exactly what makes the set of matchings of a graph an \emph{independence system}. 

\begin{definition}
	An \emph{independence system} is a pair $(U,I)$, where $U$ is an arbitrary finite set and $I\subseteq 2^U$ is a family of subsets, whose members are called the independent sets of U and satisfy:
\begin{enumerate}
\item[\emph{(i)}] $\emptyset \in I$
\item[\emph{(ii)}] If $B\in I$ and $A\subseteq B$, then $A\in I$.
\end{enumerate}
\end{definition}

Below we define a variant of Knapsack where the feasible solutions are constrained to an independence system. This forms a generalization of knapsack problems subject to matroid constraints, which are more common in the literature. \smallskip

\noindent\emph{Independence System Knapsack.} Given an independence system $(U,I)$ with costs $c_i$ and values $v_i$ on the elements of $U$, as well as a budget $B$, find $M\in I$ that maximizes  $\sum_{i\in M} v_i$ subject to $\sum_{i\in M} c_i \le B$.\smallskip

Note that for plain Knapsack $U=[n], I=2^{[n]}$, while for Budgeted Max Weighted Matching $U$ is the set of edges of a given graph $G$ and $I$ is the set of all matchings of $G$. There exist several other problems that are special cases of Independence System Knapsack, like
\begin{itemize}
	\item \emph{Budgeted Max Weighted Forest} where $U$ is the set of edges of a given graph $G$ and $I$ is the set of  acyclic subgraphs of $G$,\smallskip
	\item \emph{Budgeted Max Weighted Matroid Member} where $(U, I)$ is a matroid\footnote{A \emph{matroid} $(U, I)$ is an independence system that also has the \emph{exchange property}:  \emph{If $A, B \in I$ and $|A|<|B|$, then there exists $x\in B\mysetminus A$ such that $A\cup\{x\}\in I$.}} (Budgeted Max Weighted Forest is a special case of this problem),\smallskip
	\item \emph{Budgeted Max Independent Set} where $U$ is the set of vertices of a given graph $G$ and $I$ is the set of independence sets of $G$, and \smallskip
	\item \emph{Budgeted Max Weighted $k$-D-Matching} where $U$ is the set of hyperedges of a $k$-uniform $k$-partite hypergraph $H$ and $I$ is the set of all $k$-dimensional matchings of $H$.
\end{itemize}

\noindent The following can be easily derived as in the case of Budgeted Max Weighted Matching.
\begin{proposition}
\label{lem:ISK_in_XOS}
Every problem that can be formulated as an Independence System Knapsack problem belongs to the class XOS.
\end{proposition}

Clearly it is not always possible to find an optimal solution to Independence System Knapsack  in polynomial time, even if we remove the budget constraint. Putting the running time aside, however, \ref{fig:alg-4} combined with an exact algorithm $f$ for the problem makes \ref{fig:alg-5} (resp. \ref{fig:alg-6}) a $3$-approximate randomized (resp. $4$-approximate deterministic) truthful, individually rational, budget feasible mechanism. 

Moreover, when the unbudgeted underlying problem is easy---as is the case for Max Weighted Matching, Max Weighted Forest, and Max Weighted Matroid Member---the mechanisms run in polynomial time. Even if the unbudgeted underlying problem is $NP$-hard, as long as there is a polynomial time $\rho(n)$-approximation we get $O(\rho(n))$-approximate, truthful, individually rational, budget feasible mechanisms, e.g., for Budgeted Max Weighted $k$-D-Matching this translates to a $O(k)$-approximation mechanism. Here, $n$ is the size of the input, and we should mention that the independent sets of $U$ may not be explicitly given. Typically we assume an \emph{independence oracle} that decides for any $X\subseteq U$ whether $X\in I$. However, note that in most of the cases of Independence System Knapsack mentioned above (with the exception of Budgeted Max Weighted Matroid Member) we are given a combinatorial, succinct representation of $I$ and therefore there is no need to assume access to an oracle.

When using a $\rho(n)$-approximation algorithm we should adjust the probabilities in \ref{fig:alg-5}, namely we should use $\frac{2\rho(n)}{2\rho(n)+1}$ instead of $2/3$ and $\frac{1}{2\rho(n)+1}$ instead of $1/3$. Moreover, for both mechanisms and without loss of generality, we assume that for every $i\in U$ we have $\{i\}\in I$, or else $i$ can be excluded from the initial set $A$ of active elements that is given as input to the mechanisms.
\begin{theorem}\label{thm:ISK}
If  a deterministic $\rho(n)$-approximation algorithm $f$ for the unbudgeted version of Independence System Knapsack is given as an auxiliary input to \ref{fig:alg-4}, then \ref{fig:alg-5} (resp. \ref{fig:alg-6}) becomes a $(2\rho(n)+1)$-approximate randomized (resp. $(2\rho(n)+2)$-approximate deterministic) truthful, individually rational, budget feasible mechanism. Moreover, if $f$ runs in polynomial time so do the mechanisms.
\end{theorem}

The proof of Theorem \ref{thm:ISK} follows closely the analysis of subsection \ref{sub:matching}, so we only give a sketch highlighting the differences.

\begin{sketchproof}
The proof of truthfulness, individual rationality, and budget feasibility is exactly the same with the proofs of Lemmata \ref{lem:gutru} and \ref{lem:gubf}, if we replace ``matching'' with ``independent set of $I$'' and ``edge'' with ``element of $U$''. The proof of the approximation ratios follows closely the proofs of Lemma \ref{lem:guapx} and Theorems \ref{them:ran} and \ref{thm:det}, so we will only focus on the differences. 

Let $M^*\in I$ be an optimal budget feasible independent set  and $A\subseteq U$ be the set of active elements  at the final iteration $j$ of \ref{fig:alg-4} when the $\rho(n)$-approximate solution $M$ was returned. Also, let $M_A$ be an optimal budget feasible independent set using only elements of $A$. We have that $ v(M^*)=v(M^* \cap A)+v(M^* \mysetminus A)$. But $M^* \cap A\subseteq M^*$ is an independent set with elements from $A$ so, the analog of \eqref{eq:f} is now
\begin{equation*}
v(M^* \cap A) \leq v(M_A) \le \rho(n)\cdot v(M)\,.
\end{equation*} 
Similarly, the analog of \eqref{eq:fi} is now 
\begin{equation*}
v(M^* \mysetminus A) \leq \rho(n)\cdot v(M)+v_{i^*}\,,
\end{equation*} 
and thus we get $ v(M^*)=v(M^* \cap A)+v(M^* \mysetminus A) \leq 2\rho(n)\cdot v(M)+ v_{i^*}$, or equivalently $v(M) \geq \frac{1}{2\rho(n)}(v(M^*)-v_{i^*}) $.
Now the approximation ratio of \ref{fig:alg-5} is straightforward since the expected value of the output of the mechanism is at least
\[\frac{2\rho(n)}{2\rho(n)+1} \cdot \frac{1}{2\rho(n)}(v(M^*) - v_{i^*}) +  \frac{1}{2\rho(n)+1}v_{i^*}  = \frac{1}{2\rho(n)+1} v(M^*)\,. \]
For the approximation ratio of \ref{fig:alg-6}, using the notation of the proof of Theorem \ref{thm:det} we have $\opt(A)\leq \opt(A\mysetminus \{i^*\})+v_{i^*}  \leq 2\rho(n) v_G+ v_{i'}+v_{i^*} \leq 2 \rho(n) v_G + 2v_{i^*}$. 

If $v_{i^*} \geq v_G$ then $\opt \leq (2\rho(n) +2) v_{i^*}$, while if $v_{i^*} < v_G$ then $\opt \leq  (2\rho(n) +2) v_G $. Thus in any case we have $\opt\leq (2\rho(n) +2)v(M)$.
\end{sketchproof}

Combining Theorem \ref{thm:ISK} with the polynomial time $(k-1)$-approximation algorithm of Chan and Lau \cite{ChanL12} for Max Weighted $k$-D-Matching,
and the fact that 
Max Weighted Forest  and 
Max Weighted Matroid Member (given a polynomial time independence oracle) can be solved in polynomial time (see, e.g., \cite{Cook98}), we get the following corollary.

\begin{corollary}
We can obtain\\
\emph{(i)} randomized $3$-approximation mechanisms and deterministic $4$-approximation mechanisms for Budgeted Max Weighted Forest and Budgeted Max Weighted Matroid Member (as well as Knapsack and  Budgeted Max Weighted Matching) that run in polynomial time.\\
\emph{(ii)} randomized $3$-approximation mechanisms and deterministic $4$-approximation mechanisms for Budgeted Max Weighted Independent Set and Budgeted Max Weighted $k$-D-Matching.\\
\emph{(iii)} for any $k\ge 3$, a randomized $(2k-1)$-approximation mechanism and a deterministic $2k$-approxi\-mation mechanism for Budgeted Max Weighted $k$-D-Matching that run in polynomial time.
\end{corollary}

\begin{rem}
Max Weighted Independent Set and Max Weighted $k$-D-Matching are not submodular, as was the case for Max Weighted Matching.  Max Weighted Matroid Member (and thus Max Weighted Forest), on the other hand, is submodular and therefore the results of \cite{ChenGL11} apply. However, our approach  significantly improves both  the approximation ratio and  the running time.
\end{rem}

\begin{rem}
Naturally, Remark \ref{rem:lower} applies here as well. For every problem stated in this section there is no deterministic (resp. randomized) truthful, budget feasible mechanism with better approximation ratio than $1+\sqrt{2}$ (resp. $2$). These lower bounds are independent of any complexity assumption.
\end{rem}

\section{Conclusions}
\label{sec:concl}

We have studied further the problem of designing truthful and budget feasible mechanisms for budgeted versions of well known optimization problems. Especially for the XOS problems we considered, only randomized mechanisms with very high approximation ratios were known prior to our result. 
There are still many interesting open problems that are worth further exploration in the context of budgeted mechanism design. First, for the case of submodular functions, even though we do have a better understanding for designing mechanisms given all the previous works, the current results are still not known to be tight. We also want to stress that the literature has mostly considered nondecreasing submodular functions. Dobzinski et al.~\cite{DobzinskiPS11} gave a constant approximation mechanism for the Budgeted Max Cut problem, however it remains a very interesting problem for future work to obtain mechanisms for general non-monotone submodular valuations. Furthermore, for the XOS class, the picture is way more challenging. We would like to identify more problems that admit better approximation guarantees, even with exponential time mechanisms. A component that seems to be missing at the moment is a characterization of truthful and budget feasible mechanisms. We believe that obtaining characterization results would be crucial in resolving the above questions.



\bibliographystyle{plain}
\bibliography{budgetedMechanismDesign}

\end{document}